\documentclass[reqno,centertags,11pt]{amsart}
\usepackage{amsmath,amsthm,amscd,amssymb}
\usepackage{latexsym}
\usepackage{graphicx,psfrag}
\usepackage{pstricks}
\usepackage{pst-node}
\usepackage {verbatim}

\newcommand{\R}{{\mathbb{R}}}

\newcommand{\ol}{\overline}

\newcommand{\wti}{\widetilde  }

\newcommand{\hatt}{\widehat}
\newcommand{\beq}{\begin{equation}}
\newcommand{\eeq}{\end{equation}}
\newcommand{\bdm}{\begin{displaymath}}
\newcommand{\edm}{\end{displaymath}}
\newcommand{\ba}{\begin{align}}
\newcommand{\ea}{\end{align}}
\newcommand{\bpf}{\begin{proof}}
\newcommand{\epf}{\end{proof}}

\newcommand{\la}{\langle}
\newcommand{\ra}{\rangle}

\newcommand{\supp}{\mathrm{supp}\, }               
\newcommand{\dist}{\mathrm{dist}}               

\newcommand{\veps}{\varepsilon}

\allowdisplaybreaks

\newcommand{\calQ}{\mathcal{Q}}

\newtheorem{thm}{Theorem}
\newtheorem{prop}[thm]{Proposition}
\newtheorem{lem}[thm]{Lemma}
\newtheorem{cor}[thm]{Corollary}

\theoremstyle{definition}

\newtheorem{remark}[thm]{Remark}

\newcounter{theoremi}[thm]

\numberwithin{thm}{section}
\numberwithin{equation}{section}

\begin{document}

\title[Exponential decay of dispersion managed solitons]{Exponential decay of dispersion managed
solitons for vanishing average dispersion}
\author[M.~B.~Erdo\smash{\u g}an, D.~Hundertmark, and Y.-R.~Lee]{M.~Burak Erdo\smash{\u g}an, Dirk Hundertmark,
and Young-Ran~Lee}
\address{Department of Mathematics, Altgeld Hall, University of Illinois at Urbana-Champaign,
    1409 W.~Green Street, Urbana, IL 61801.}%
\email{berdogan@math.uiuc.edu}%
\address{Department of Mathematics, Altgeld Hall, and Institute for Condensed Matter Theory,
            University of Illinois at Urbana-Champaign, 1409 W.~Green Street, Urbana, IL 61801.}%
\email{dirk@math.uiuc.edu}%
\address{Department of Mathematics, Sogang University, Shinsu-dong 1, Mapo-gu, Seoul, 121-742, South Korea.}%
\email{younglee@sogang.ac.kr}

\thanks{\copyright 2009 by the authors. Faithful reproduction of this article,
        in its entirety, by any means is permitted for non-commercial purposes}

\begin{abstract}
We show that any $L^2$ solution of the Gabitov-Turitsyn equation describing dispersion managed solitons decay exponentially in space and frequency domains. This confirms in the affirmative Lushnikov's conjecture of exponential decay of dispersion managed solitons.
\end{abstract}

\maketitle

\section{Introduction}\label{sec1}

Consider the one-dimensional non-linear Schr\"odinger
equation (NLS) with periodically varying dispersion
coefficient
 \beq\label{eq:periodic-NLS}
  i u_t + d(t) u_{xx} +  |u|^2 u = 0.
 \eeq
It describes the amplitude of a signal
transmitted via amplitude modulation of a carrier wave through a
fiber-optical cable where the dispersion is varied periodically
along the fiber, see, e.g., \cite{Agrawal95,SulemSulem,Tetal03}.
In \eqref{eq:periodic-NLS} $t$ corresponds to
the distance along the fiber, $x$ denotes the (retarded) time,  and $d(t)$ is the dispersion along
the waveguide which, for practical purposes, one can assume to be piecewise constant.

In fiber optic cables, the information can be transmitted using localized soliton pulses in allocated time slots;
the presence of a pulse corresponds to ``1'' and the absence of a pulse corresponds to ``0''
in binary format. Solitary solutions exist due to a delicate balance between the dispersion and nonlinearity.
In order to minimize the interaction between the individual pulses, one needs to keep the pulses sufficiently
far apart.  A draw-back of solitary solutions for this application  is that  the soliton solutions with small
support have  large $L^2$ norm unless the dispersion constant is   small.
The technique of dispersion management was invented to overcome this difficulty. The
idea  is to use
alternating sections of constant but (nearly) opposite
dispersion. This introduces a rapidly varying dispersion $d(t)$ with small average dispersion,
leading to well-localized stable soliton-like pulses changing periodically along the fiber. This idea
has been enormously fruitful
(see, {\em e.g.},  \cite{LKC80,AB98, GT96a, GT96b, KH97, Kurtzge93, LK98, MM99, MMGNMG-NV99}).
Record breaking transmission rates had been achieved using this technology \cite{MGLWXK03} which is now widely used commercially.

To study strong dispersion management regime, it is convenient to write
$$d(t)=\frac{1}{\veps}d_0(t/\veps )+d_{av},$$
Here $d_0(t)$ is the mean zero part which we assume to be piecewise constant, and $d_\text{av}$
the average dispersion over one period , and $\veps$ is a small parameter.
Rescaling $t/\veps$ to $t$, the envelope equation takes the form
\begin{equation}\label{eq:GT}
  i u_t + d_0(t) u_{xx} +  \veps d_{\text{av}}u_{xx}+ \veps|u|^2 u =0.
 \end{equation}

Since the full equation \eqref{eq:GT} is very hard to study, Gabitov and Turitsyn suggested to
separate the free motion given
by the solution of $iu_t +d_0(t)u_{xx}=0$ in \eqref{eq:GT}, and  to average
over one period, see \cite{GT96a,GT96b}. In the case\footnote{In fact, our method can be extended to more general dispersion profiles. We will address this issue together with the case $d_\text{av}>0$ in a later paper.} when $d_0(t)=1$ on $[-1,0]$ and $d_0(t)=-1$ on $[0,1]$
this yields the following equation for the ``averaged'' solution $v$
 \begin{align}\label{eq:DMS-time}
  &i v_t + \veps d_{\text{av}} v_{xx} +  \veps Q(v,v,v) = 0,\quad \text{where}\\
 \label{eq:Q3}
  & Q(v_1,v_2,v_3) :=   \int_0^1 T_{r}^{-1}\Big( {T_{r}v_1} \overline{T_{r}v_2} T_{r}v_3 \Big) ds,
  \end{align}
and $T_r= e^{ir\partial_x^2}$. In some sense, $v$ is the  slowly varying part of the amplitude and the varying dispersion is interpreted as a fast background oscillation, justifying formally the above
averaging procedure. This is similar to Kapitza's treatment of the unstable pendulum, see \cite{LandauLifshitz}.
This averaging procedure yielding \eqref{eq:DMS-time} is well-supported by numerical and theoretical
studies, see,
for example, \cite{AB98,TDNMSF99,Tetal03}, and was rigorously justified in
\cite{ZGJT01} in the limit of large local dispersion, i.e., as $\veps \to 0$.

One can find stationary solutions by making the ansatz $v(t,x)= e^{i\veps\omega t} f(x)$  in
\eqref{eq:DMS-time}. This yields the time independent equation
 \begin{equation}\label{eq:GabTur}
  -\omega f = -d_{\text{av}}f_{xx} -Q(f,f,f)
 \end{equation}
describing stationary soliton-like solutions, the so-called dispersion managed
solitons. Despite the enormous interest in dispersion managed solitons, there are few rigorous results
available. One reason for this is that it is a nonlinear and nonlocal equation.
Existence and smoothness of weak solutions of \eqref{eq:GabTur} had first been rigorously established in
\cite{ZGJT01} for positive average dispersion $d_{\text{av}}>0$.
In the case $d_{\text{av}}=0$, the existence was obtained in \cite{Kunze04}, also see \cite{HL2010}
for a simplified proof.
Smoothness in the case $d_{av}=0$ was established in \cite{Stanislavova}.

\begin{remark} By a weak solution we mean
$f\in H^1$ in the case $d_{av}>0$, or $f\in L^2$ in the case $d_{av}=0$, such that
\beq
  - \omega \la g,f\ra = d_{av}\la g^\prime,f^\prime\ra-\la g, Q(f,f,f)\ra.
\eeq
 for all $g\in H^1$. Here $\la g, f\ra= \int_\R \ol{g(x)}f(x)\, dx$
 is the usual scalar product on $L^2(\R)$.

By a formal calculation, using the unicity of $T_r$ in $L^2$, we have
$$\la g, Q(f,f,f)\ra=\calQ(g,f,f,f),$$
where
\beq\label{eq:Qdef}
\calQ(f_1,f_2,f_3,f_4)= \int\limits_0^1\int\limits_\R \ol{T_{r}f_1(x)} T_{r}f_2(x)  \ol{T_{r}f_3(x)} T_{r}f_4(x)  dx ds.
\eeq
The functional $\calQ(f_1,f_2,f_3,f_4)$
 is well-defined for $f_j\in L^2(\R)$ due to Strichartz inequality, see \cite{ZGJT01,HL2008}.
\end{remark}

The decay of the solutions was first addressed by Lusnikov in \cite{lus1}.  He gave convincing but
non-rigorous arguments that any solution $f$ of \eqref{eq:GabTur} for $d_{av}=0$ satisfies
 \begin{equation}\label{eq:luspre}
  f(x)\sim |x|\cos(a_0x^2 + a_1 x + a_2) e^{-b|x|} \quad \text{as }
  x\to\infty
 \end{equation}
for some suitable choice of real constants $a_j $ and $b>0$, see
also \cite{lus2}. In particular, he predicted that $f$ and $\widehat{f}$ decay exponentially at
infinity. For $d_{av}=0$, the first rigorous $x$-space decay bounds were established in
\cite{HL2008}, where it was shown that both $f$ and $\widehat{f}$ decay faster than any polynomial
in the case $d_{\text{av}}=0$. In particular any weak solution is a Schwartz function.

Our main result confirms Lusnikov's exponential decay prediction:
\begin{thm} \label{thm:main} Assume that $d_{av}=0$.
Let $f\in L^2$ be a weak solution of \eqref{eq:GabTur}. Then there exists
$\mu>0$ such that
$$|f(x)|\lesssim e^{-\mu |x|},\quad |\widehat{f}(\xi)|\lesssim e^{-\mu |\xi|},$$
where $\widehat{f}$ is the Fourier transform of $f$.
\end{thm}
We have the following immediate corollary.
\begin{cor}\label{cor:analyticity}
Under the conditions of Theorem~\ref{thm:main}, both $f$ and $\widehat f$ are
analytic in a strip containing the real line.
\end{cor}

\begin{remark}
Weak solutions of \ref{eq:GabTur} can be found with the help of a variational principle. For $d_\mathrm{av}=0$ it is given by
  \beq \label{eq:best-constant}
    P_\lambda:= \sup\{ \calQ(f,f,f,f)\vert \, \|f\|^2_2=\lambda\}
  \eeq
 Note $\calQ(f,f,f,f)= \int_0^1\int |e^{it\partial_x^2}f(x)|^4\, dx dt =  \|e^{it\partial_x^2}f\|^4_{L^4_{[0,1]}L^4_x}$ and $e^{it\partial_x^2}f$ is the space-time Fourier transform of a measure concentrated on the paraboloid $\{\tau=k^2\}\subset\R^2$ with square-integrable density $\hat{f}$,
  \bdm
    e^{it\partial_x^2}f(x) =
      \frac{1}{2\pi}\iint e^{ixk +i t\tau} \delta(\tau-k^2) \hat{f}(k)\, d\tau dk .
  \edm
Since, by scaling, $P_\lambda= P_1\lambda^2$, the variational problem \eqref{eq:best-constant} yields the best constant in the $L^4$ Fourier extension estimate
  \beq \label{eq:extension-estimate}
    \|e^{it\partial_x^2}f\|^4_{L^4_{[0,1]}L^4_x} \le P_1\|f\|^4_{L^2(\R)} .
  \eeq
  for measure with an $L^2$ density on the paraboloid.
Existence of maximizers for the variational problem \eqref{eq:best-constant} was established in \cite{Kunze04}, see also \cite{HL2010}. Thus our Theorem \ref{thm:main} and Corollary \ref{cor:analyticity} for any weak solutions of \eqref{eq:GabTur} for vanishing average dispersion show, in particular, strong regularity properties for any maximizer of the Fourier extension estimate \eqref{eq:extension-estimate}.

The inequality \eqref{eq:extension-estimate} is, of course, closely related to the one-dimensional Strichartz inequality
 \beq\label{eq:Strichartz}
   \|e^{it\partial_x^2}f\|_{L^6_tL^6_x} \le S_1\|f\|_{L^2(\R)} ,
 \eeq
for which the sharp constant and the maximizers have been classified in \cite{BBCH09,Foschi,HunZha}, and the Fourier extension problem for the sphere for which the existence of maximizers and their properties has recently been discussed in \cite{CS2010}.
\end{remark}

In the proof of Theorem~\ref{thm:main}, the central idea is, as in \cite{Agm, HunSig}, to obtain suitable exponentially weighted a-priori estimates for the weak solution.
In \cite{Agm, HunSig}, the commutator of the exponential weight with the Schr\"odinger operator is easily calculated since the operator is local. Variations of Agmon's method work also for relativistic Schr\"odinger operators which are nonlocal. However, in these applications one relies on the pointwise decay of the corresponding kernel. In our case, a major difficulty arises since our operator $Q$ is nonlocal and the kernel of the free Schr\"odinger evolution has no pointwise decay.
We overcome this difficulty by using the multi-linear structure
of $\calQ$ and the oscillation of the kernel. This is done in Section~\ref{sec:preparations}, where we obtain exponentially weighted multi-linear estimates for $\calQ$. Multi-linear refinements of the Strichartz estimate where first established in \cite{Bourgain1998} and later systematically studied in \cite{Tao2001}. The results of these two papers focus, however, on the Fourier side and do not allow exponential weights. More importantly, we require bounds \emph{independent} of the exponential weights, see Theorems \ref{thm:twistedQ} and \ref{thm:twistedQ-loc} below.
Our bounds are refinements of the $x$-space Strichartz estimates which were developed in \cite{HL2008} and used in conjunction with well-known Fourier space Strichartz estimates to prove that any weak solution is a Schwartz function. We would also like to note that our proof  uses only the fact that $f\in L^2$ and as such does not require any of the previous smoothness or decay results.

\section{A-priori estimates for $\calQ$} \label{sec:preparations}

We start with two alternative representations of $\calQ$ inspired by the
calculations in \cite{HunZha}.
\begin{lem}\label{lem:repr}
\begin{align}\label{eq:Qalt1}
&\calQ(f_1,f_2,f_3,f_4)=\\
&\quad \frac{1}{4\pi}\int\limits_{0}^1 \int\limits_{\R^4} \frac{1}{t} e^{-i(\eta_1^2-\eta_2^2 +\eta_3^2 -\eta_4^2)/(4t)}
      \ol{f_1(\eta_1)} f_2(\eta_2)\ol{f_3(\eta_3)} f_4(\eta_4)
        \delta(\eta_1 - \eta_2 + \eta_3 - \eta_4)\, d\eta d t  \nonumber
\\\label{eq:Qalt2}
&\calQ(f_1,f_2,f_3,f_4)=\\
&\quad\frac{1}{2\pi}\int\limits_{0}^1 \int\limits_{\R^4}  e^{it(\eta_1^2-\eta_2^2 +\eta_3^2 -\eta_4^2) }
      \ol{\widehat{f_1}(\eta_1)} \widehat{f_2}(\eta_2)\ol{\widehat{f_3}(\eta_3)} \widehat{f_4}(\eta_4)
        \delta(\eta_1 - \eta_2 + \eta_3 - \eta_4)\, d\eta d t \nonumber
\end{align}
\end{lem}
\begin{proof}
Using the formula
$$T_tf(x)=\frac{1}{\sqrt{4\pi i t}}\int\limits_\R e^{i(x-y)^2/(4t)}f(y) dy,$$
we get
\begin{align*}
&\ol{T_tf_1(x)} T_tf_2(x)  \ol{T_tf_3(x)} T_tf_4(x)=\\
&\frac{1}{(4\pi  t)^2} \int\limits_{\R^4} e^{ix(\eta_1-\eta_2+\eta_3-\eta_4)/(2t)}e^{-i(\eta_1^2-\eta_2^2 +\eta_3^2 -\eta_4^2)/(4t)}\ol{f_1(\eta_1)} f_2(\eta_2)\ol{f_3(\eta_3)} f_4(\eta_4)  d\eta.
\end{align*}
From which one obtains \eqref{eq:Qalt1} by performing the $x$-integration.

Similarly, one obtains \eqref{eq:Qalt2} by using the formula
$$T_tf(x)=\frac{1}{\sqrt{2\pi }}\int\limits_\R e^{ix\xi}e^{-it\xi^2}\hatt{f}(\xi) d\xi.$$
\end{proof}

To obtain exponential decay of dispersion managed solitons we need the following
`twisted' dispersion management functionals
 \begin{align*}
    \calQ_{\mu,\veps}(h_1,h_2,h_3,h_4)
    & :=
        \calQ(e^{F_{\mu,\veps}(X)}h_1,e^{-F_{\mu,\veps}(X)}h_2,e^{-F_{\mu,\veps}(X)}h_3,e^{-F_{\mu,\veps}(X)}h_4) \\
    \wti{\calQ}_{\mu,\veps} (h_1,h_2,h_3,h_4)
    &:=
        \calQ(e^{F_{\mu,\veps}(P)}h_1,e^{-F_{\mu,\veps}(P)}h_2 ,e^{-F_{\mu,\veps}(P)}h_3 ,e^{-F_{\mu,\veps}(P)}h_4).
 \end{align*}
Here $X$ denotes multiplication by $x$ and $P=-i\partial_x$ is the one-dimensional momentum operator, and
 \beq \label{eq:Fmueps}
    F_{\mu,\veps}(x):= \mu \frac{|x|}{1+\veps|x|}, \quad \mu,\veps \ge 0.
 \eeq
We have the following theorems  which are rather surprising at first sight. They are the basis for our proof of exponential decay of dispersion management solitons.

\begin{thm}\label{thm:twistedQ} There exists a constant $C$ such that the bounds
 \begin{align}
   | \calQ_{\mu,\veps}(h_1,h_2,h_3,h_4)| &\le C \prod_{j=1}^4 \|h_j\|_2 ,\\
   | \wti{\calQ}_{\mu,\veps}(h_1,h_2,h_3,h_4)| &\le C \prod_{j=1}^4 \|h_j\|_2
 \end{align}
 hold for all $\mu,\veps\ge 0$.
\end{thm}

\begin{thm}\label{thm:twistedQ-loc}
There exists a constant $C$ such that if for some $l,k\in\{1,2,3,4\}$
 $\tau=\dist(\supp(h_l), \supp(h_k)) \ge 1$  then
 \beq
    |\calQ_{\mu,\veps}(h_1,h_2,h_3,h_4)|\le \frac{C}{\sqrt{\tau}} \prod_{j=1}^4 \|h_j\|_2
 \eeq
 for all $\mu,\veps\ge 0$. Moreover, if $\tau=\dist(\supp(\hatt{h_l}), \supp(\hatt{h_k}))\ge 1$
 then also
 \beq
    |\wti{\calQ}_{\mu,\veps}(h_1,h_2,h_3,h_4)|\le \frac{C}{\sqrt{\tau}} \prod_{j=1}^4 \|h_j\|_2
 \eeq
\end{thm}
\begin{remark}
    The point of Theorems \ref{thm:twistedQ} and \ref{thm:twistedQ-loc} is that the constant in the bounds
    is \emph{independent} of $\mu,\veps\ge 0$. We explicitly allow $\veps=0$, which at first seems to be in
    conflict with the fact that we need $e^{F_{\mu,0}}h_1\in L^2$. However, in this case we can restrict
    ourselves to compactly supported functions $h_1$ and then use the a-priori bound and the density of
    these functions in $L^2$.
\end{remark}

Let $M$ be a multiplier in the variables $\eta_1, \eta_2,\eta_3,\eta_4$ and define the
oscillatory functionals
\begin{align*}
    K^{1}_M &(h_1,h_2,h_3,h_4):= \\
        &\int\limits_{0}^1 \int\limits_{\R^4} \frac{1}{t} e^{-i(\eta_1^2-\eta_2^2 +\eta_3^2 -\eta_4^2)/(4t)}
        M(\eta) \ol{h_1(\eta_1)} h_2(\eta_2)\ol{h_3(\eta_3)} h_4(\eta_4)
        \delta(\eta_1 - \eta_2 + \eta_3 - \eta_4) \, d\eta d t  \\
    K^{2}_M &(h_1,h_2,h_3,h_4) := \\
        &\int\limits_{0}^1 \int\limits_{\R^4} e^{it(\eta_1^2-\eta_2^2 +\eta_3^2 -\eta_4^2)}
        M(\eta ) \ol{h_1(\eta_1)} h_2(\eta_2)\ol{h_3(\eta_3)} h_4(\eta_4)
        \delta(\eta_1 - \eta_2 + \eta_3 - \eta_4)\, d\eta d t
\end{align*}
Note that by Lemma~\ref{lem:repr}, we can rewrite the twisted functionals as
\begin{align*}
    \calQ_{\mu,\veps} (h_1,h_2,h_3,h_4)&= \frac{1}{4\pi}K^{1}_{M_{\mu,\veps}}(h_1,h_2,h_3,h_4), \\
    \wti{\calQ}_{\mu,\veps}(h_1,h_2,h_3,h_4) & =\frac{1}{2\pi} K^{2}_{M_{\mu,\veps}}(\hatt{h_1},\hatt{h_2},\hatt{h_3},\hatt{h_4}),
\end{align*}
where
$$M_{\mu,\veps}(\eta)=e^{F_{\mu,\veps}(\eta_1) -F_{\mu,\veps}(\eta_2) -F_{\mu,\veps}(\eta_3) -F_{\mu,\veps}(\eta_4)}.$$
Note that by the triangle inequality the function $M_{\mu,\veps}$  is bounded by $1$ on
the set $\eta_1-\eta_2+\eta_3-\eta_4=0$ for any $\mu, \veps\ge 0$.
Therefore Theorems \ref{thm:twistedQ} and \ref{thm:twistedQ-loc} follow
immediately from the Propositions \ref{prop:KM} and \ref{prop:KM-loc} below.
\begin{prop}\label{prop:KM}
Let $\wti{M}:= \sup_{\eta_1-\eta_2+\eta_3-\eta_4=0} M(\eta_1,\eta_2,\eta_3,\eta_4)<\infty$. Then
$K^{n}_M$, $n=1,2$, is well-defined for all $h_j\in L^2(\R)$. Moreover,
 \beq
    |K^{n}_M (h_1,h_2,h_3,h_4)|
    \lesssim
    \wti{M} \prod_{j=1}^4 \|h_j\|_2
 \eeq
 where the implicit constant is independent of $M$ and $h_j$, $j=1,2,3,4$.
\end{prop}
\begin{proof} By scaling, we can assume $\wti{M}=1$.
Let $a(\eta):=\eta_1^2-\eta_2^2 +\eta_3^2 -\eta_4^2$. We write
\begin{align}
|K^{1}_M| &\leq  \int\limits_{\R^4} \Big|\int\limits_{0}^1\frac{1}{t} e^{-ia(\eta)/(4t)} dt \Big|
        |M(\eta)| \prod_{j=1}^4 |h_j(\eta_j)|
        \delta(\eta_1 - \eta_2 + \eta_3 - \eta_4)    d\eta \nonumber\\
&\leq  \wti{M} \int\limits_{\R^4} \Big|\int\limits_{0}^1\frac{1}{t} e^{-ia(\eta)/(4t)} dt \Big|
   \prod_{j=1}^4 |h_j(\eta_j)|
        \delta(\eta_1 - \eta_2 + \eta_3 - \eta_4)    d\eta. \label{eq:K1M2}
\end{align}
Now we divide the $t$-integral into two pieces $t\leq |a(\eta)|$ where oscillations will be important
and $t\geq |a(\eta)|$. More precisely,
\begin{align}
\eqref{eq:K1M2} & \leq   \int\limits_{\R^4} \Big|\int\limits_{0}^{\min(1,|a(\eta)|)}\frac{e^{-ia(\eta)/(4t)} }{t} dt \Big|
         \prod_{j=1}^4 |h_j(\eta_j)|
        \delta(\eta_1 - \eta_2 + \eta_3 - \eta_4)    d\eta \label{eq:I1}\\
&+   \int\limits_{\R^4}  \int\limits_{\min(1,|a(\eta)|)}^1\frac{dt}{t}
   \prod_{j=1}^4 |h_j(\eta_j)|
        \delta(\eta_1 - \eta_2 + \eta_3 - \eta_4)    d\eta. \label{eq:I2}
\end{align}
Let us introduce the following functionals, which, for later flexibility, we define in a little bit
more generality than needed at the moment. For any (measurable) subset $A\subset\R^4$ let
\begin{align}
I_1(A) & :=  \int\limits_{A} \min(1, |a(\eta)|^{-1}) \prod_{j=1}^4 |h_j(\eta_j)|
                \delta(\eta_1 - \eta_2 + \eta_3 - \eta_4)    d\eta \label{eq:I1-loc}\\
I_2(A) & :=   \int\limits_0^1 \int\limits_{A\cap\{ |a(\eta)|\le t \}} \frac{1}{t}
                \prod_{j=1}^4 |h_j(\eta_j)|
                \delta(\eta_1 - \eta_2 + \eta_3 - \eta_4)    d\eta dt. \label{eq:I2-loc}
\end{align}
By Fubini-Tonelli
\begin{align*}
   \eqref{eq:I2} = &
   \int\limits_0^1 \frac{1}{t} \int\limits_{|a(\eta)|\leq t}
   \prod_{j=1}^4 |h_j(\eta_j)|
        \delta(\eta_1 - \eta_2 + \eta_3 - \eta_4)    d\eta dt =  I_2(\R^4).
\end{align*}
To estimate \eqref{eq:I1} we employ the following bound, which follows
by the change of variable $\tau= 1/t$ and then an
integration by parts.
 \begin{align} \label{eq:IBP}
\Big|\int\limits_{0}^{b}   \frac{1}{t} e^{-ia/(4t)} d t
\Big| \leq  \frac{8 |b|}{|a|}.
\end{align}
Using this one sees
 \bdm
    \eqref{eq:I1} \lesssim I_1(\R^4),
 \edm
hence
 \beq\label{eq:KMbound}
 |K^{1}_M|  \lesssim  I_1(\R^4) +I_2(\R^4)
 \eeq
We start to estimate the second term. Since $|a(\eta)|=2(\eta_1-\eta_2)(\eta_2-\eta_3)$
on the set $\eta_1 - \eta_2 + \eta_3 - \eta_4=0$, we
can estimate the $\eta$-integral in $I_2(\R^4)$ by
\begin{equation}\label{eq:eta-int}
\begin{split}
   \int\limits_{|a(\eta)|\leq t}
   \prod_{j=1}^4 |h_j(\eta_j)|
        &\delta(\eta_1 - \eta_2 + \eta_3 - \eta_4)    d\eta \leq \\
&\int\limits_{|\eta_1-\eta_2|\leq \sqrt{t}}
   \prod_{j=1}^4 |h_j(\eta_j)|
        \delta(\eta_1 - \eta_2 + \eta_3 - \eta_4)    d\eta\\
&+
\int\limits_{|\eta_2-\eta_3|\leq \sqrt{t}}
   \prod_{j=1}^4 |h_j(\eta_j)|
        \delta(\eta_1 - \eta_2 + \eta_3 - \eta_4)    d\eta
\end{split}
\end{equation}
The first integral on the right hand side of \eqref{eq:eta-int} can be bounded by
\begin{align*}
 \int\limits_{|\eta_1-\eta_2|\leq \sqrt{t}}   &
   |h_1(\eta_1)| |h_2(\eta_2)| |h_3(\eta_3)| |h_4(\eta_1-\eta_2+\eta_3)| d\eta_1 d\eta_2 d\eta_3\\
&\leq \|h_3\|_2\|h_4\|_2
 \int\limits_{|\eta_1-\eta_2|\leq \sqrt{t}}
   |h_1(\eta_1)| |h_2(\eta_2)|   d\eta_1 d\eta_2  \\
&\leq \|h_3\|_2\|h_4\|_2 \Big(\int\limits_{|\eta_1-\eta_2|\leq \sqrt{t}}
   |h_1(\eta_1)|^2   d\eta_1 d\eta_2\Big)^{1/2}\Big(\int\limits_{|\eta_1-\eta_2|\leq \sqrt{t}}
   |h_2(\eta_2)|^2   d\eta_1 d\eta_2\Big)^{1/2}\\
&=2\sqrt{t}\|h_1\|_2\|h_2\|_2\|h_3\|_2\|h_4\|_2,
\end{align*}
where we used Cauchy-Schwarz inequality first in the $d \eta_3$ integral, then in the $d\eta_1 d\eta_2$
integral. The second integral can be estimated similarly. Thus
\begin{align*}
 I_2(\R^4) \leq 4   \|h_1\|_2\|h_2\|_2\|h_3\|_2\|h_4\|_2 \int\limits_0^1 \frac{1}{\sqrt{t}} dt
 = 8   \|h_1\|_2\|h_2\|_2\|h_3\|_2\|h_4\|_2.
\end{align*}
To estimate $I_1(\R^4)$ we split the $\eta$-integral into two disjoint regions
\begin{align*}
A_1 :&\!= \{\eta\in\R^4: |\eta_1-\eta_2|\leq 1 \text{ or } |\eta_2-\eta_3|\leq 1\} \\
    &= \{\eta\in\R^4: |\eta_1-\eta_2|\leq 1 \}\cup \{\eta\in\R^4: |\eta_2-\eta_3|\leq 1\}\\
A_2 :&\!= \{\eta\in\R^4: |\eta_1-\eta_2|> 1 \text{ and } |\eta_2-\eta_3|> 1\}.
\end{align*}
Obviously,  $I_1(\R^4)= I_1(A_1)+ I_1(A_2)$. For $I_1(A_1)$, we bound the minimum
in \eqref{eq:I1-loc} by $1$ and then estimate the remaining
integral as in \eqref{eq:eta-int} but now for  $t=1$. This shows
 \bdm
    I_1(A_1) \lesssim  \|h_1\|_2 \|h_2\|_2\|h_3\|_2\|h_4\|_2 .
 \edm
On the other hand
\begin{equation}
    \begin{split}\label{eq:I1A2}
 I_1(A_2)
        & \le    \int\limits_{A_2}\frac{ |h_1(\eta_1)| |h_2(\eta_2)| |h_3(\eta_3)| |h_4(\eta_4)|}{|a(\eta)|}
            \delta(\eta_1-\eta_2+\eta_3-\eta_4) d\eta \\
        & \lesssim    \int\limits_{{|\eta_1-\eta_2|\ge 1}\atop{|\eta_2-\eta_3|\ge 1}}
            \frac{ |h_1(\eta_1)| |h_2(\eta_2)| |h_3(\eta_3)| |h_4(\eta_1-\eta_2+\eta_3)|}{|\eta_1-\eta_2|  |\eta_2-\eta_3|}
            d\eta_1 d\eta_2 d\eta_3\\
        &\leq   \|h_1\|_2\|h_2\|_2 \|h_4\|_2
            \Big(\int\limits_{{|\eta_1-\eta_2|\ge 1}\atop{|\eta_2-\eta_3|\ge 1}}
            \frac{   |h_3(\eta_3)|^2  }{|\eta_1-\eta_2|^2  |\eta_2-\eta_3|^2}
            d\eta_1 d\eta_2 d\eta_3\Big)^{1/2}\\
        &\lesssim   \|h_1\|_2\|h_2\|_2 \|h_3\|_2\|h_4\|_2.
    \end{split}
\end{equation}
This finishes the proof for $K_M^1$. The proof for $K_M^2$ is simpler. Using the inequality
 \beq\label{eq:easiIBP}
    \Big|\int\limits_0^1 e^{iat}dt\Big|\lesssim\min(1,|a|^{-1}),
 \eeq
one realizes $|K_M^2|\lesssim I_1(\R^4) $ and then proceeds as in the bound of $I_1(\R^4)$.
\end{proof}

A refinement of this proposition, when at least two of the functions, say, $h_j$ and $h_k$, have separated
supports, is
\begin{prop}\label{prop:KM-loc}
 Assume that $\wti{M}:= \sup_{\eta_1-\eta_2+\eta_3-\eta_4=0} M(\eta_1,\eta_2,\eta_3,\eta_4)<\infty$ and that
 there exist $l,k\in \{1,2,3,4\}$ with $\tau=\dist(\supp(h_l), \supp(h_k))\geq 1$. Then,
 \beq
    |K^{n}_M (h_1,h_2,h_3,h_4)|
    \lesssim
    \frac{\wti{M}}{\sqrt{\tau}} \prod_{j=1}^4 \|h_j\|_2,\,\,\,\,\,\,n=1,2.
 \eeq
 \end{prop}
\begin{proof}
Again we can assume $\wti{M}=1$.
For $A\subset\R^4$ let
 \beq\label{eq:I-loc}
    I(A):= \int\limits_{A} \Big|\int\limits_{0}^1\frac{1}{t} e^{-ia(\eta)/(4t)} dt \Big|
   \prod_{j=1}^4 |h_j(\eta_j)|
        \delta(\eta_1 - \eta_2 + \eta_3 - \eta_4)    d\eta.
 \eeq
Let $J_{l,k}^\tau:=\{\eta\in \R^4:|\eta_l-\eta_k|\geq\tau\}$.
Then $|K^1_M |\le I(J_{l,k}^\tau)$. By symmetry in \eqref{eq:I-loc}, it is enough to consider the cases  $(l,k)=(1,2)$ and $(l,k)=(1,3)$. First we consider the case $(l,k)=(1,2)$.

Recalling the definitions \eqref{eq:I1-loc} and \eqref{eq:I2-loc},
we can, as in the proof of Proposition~\ref{prop:KM}, bound $I(J_{1,2}^\tau)$ by
 \beq
    I(J_{1,2}^\tau) \lesssim I_1(J_{1,2}^\tau) + I_2(J_{1,2}^\tau)
 \eeq
In the integral defining  $I_2(J_{1,2}^\tau)$, we have
$t\geq |a(\eta)| =2|\eta_1-\eta_2||\eta_2-\eta_3|\geq 2\tau|\eta_2-\eta_3|$,
which implies $|\eta_2-\eta_3|\lesssim t/\tau$. This yields
\begin{align*}
I_2(J_{1,2}^\tau)&\lesssim \int\limits_0^1 \frac{1}{t}\int\limits_{|\eta_2-\eta_3|\lesssim t/\tau}\prod_{j=1}^4 |h_j(\eta_j)|
        \delta(\eta_1 - \eta_2 + \eta_3 - \eta_4)    d\eta dt\\
&\lesssim  \frac{1}{\tau}  \|h_1\|_2 \|h_2\|_2\|h_3\|_2\|h_4\|_2,
\end{align*}
where we obtained the last line as in the estimate of \eqref{eq:eta-int} with $\sqrt{t}$ replaced by $t/\tau$.

To estimate $I_1(J_{1,2}^\tau)$ let
\begin{align}
    A_1^\tau & :=  \{|\eta_1-\eta_2|\ge \tau\}\cap\{ |\eta_2-\eta_3| \le 1\} \\
    A_2^\tau & := \{|\eta_1-\eta_2|\ge \tau\}\cap\{ |\eta_2-\eta_3|> 1\}  .
\end{align}
Then, obviously,
 \bdm
    I_1(J_{1,2}^\tau) = I_1(A_1^\tau)+I_1(A_2^\tau) .
 \edm
Similar to \eqref{eq:I1A2}, we bound $I_1(A_2^\tau)$ by
\begin{align*}
        I_1(A_2^\tau)
    &\lesssim
        \|h_1\|_2\|h_2\|_2 \|h_4\|_2
            \Big(\int\limits_{{|\eta_1-\eta_2|\ge \tau}\atop{|\eta_2-\eta_3|\ge 1}}
            \frac{   |h_3(\eta_3)|^2  }{|\eta_1-\eta_2|^2  |\eta_2-\eta_3|^2}
            d\eta_1 d\eta_2 d\eta_3\Big)^{1/2} \\
    &\lesssim
        \frac{1}{\sqrt{\tau}} \|h_1\|_2\|h_2\|_2 \|h_3\|_2\|h_4\|_2
\end{align*}
For estimating $I_1(A_1^\tau)$ we bound the minimum in \eqref{eq:I1-loc} by $|a(\eta)|^{-1/2}$
 to see
 \begin{align*}
        I_1(A_1^\tau)
    &\lesssim
        \int\limits_{A_1^\tau} \frac{1}{|\eta_1-\eta_2|^{1/2}|\eta_2-\eta_3|^{1/2}}
         \prod_{j=1}^4 |h_j(\eta_j)|
        \delta(\eta_1 - \eta_2 + \eta_3 - \eta_4)    d\eta \\
    &\lesssim
        \frac{1}{\sqrt{\tau}} \int\limits_{|\eta_2-\eta_3|\le 1}
            \frac{|h_1(\eta_1)h_2(\eta_2)h_3(\eta_3)h_4(\eta_1-\eta_2+\eta_3)|}{|\eta_2-\eta_3|^{1/2}}
        d\eta_1 d\eta_2 d\eta_3 \\
    &\le
        \frac{\|h_1\|_2\|h_4\|_2}{\sqrt{\tau}}
        \int\limits_{|\eta_2-\eta_3|\le 1} \frac{|h_2(\eta_2)h_3(\eta_3)|}{|\eta_2-\eta_3|^{1/2}} d\eta_2 d\eta_3 \\
    &\le
        \frac{\|h_1\|_2\|h_4\|_2}{\sqrt{\tau}}
        \Big(\int\limits_{|\eta_2-\eta_3|\le 1} \frac{|h_3(\eta_3)|^2 d\eta_2 d\eta_3}{|\eta_2-\eta_3|^{1/2}}  \Big)^{1/2}
        \Big(\int\limits_{|\eta_2-\eta_3|\le 1} \frac{|h_2(\eta_2)|^2 d\eta_2 d\eta_3}{|\eta_2-\eta_3|^{1/2}}  \Big)^{1/2} \\
    &\lesssim
        \frac{1}{\sqrt{\tau}} \|h_1\|_2\|h_2\|_2\|h_3\|_2\|h_4\|_2,
 \end{align*}
where in the third inequality we used the Cauchy Schwarz bound with respect to
$d\eta_1$ and in the forth inequality with respect to the measure
$|\eta_2-\eta_3|^{-{1/2}} d\eta_2 d\eta_3$. This finishes the proof for $K^1_M$.

Again the proof for $K_M^2$ is simpler. Using \eqref{eq:easiIBP} and the separation condition
$|\eta_1-\eta_2|\ge \tau$ for all $(\eta_1,\eta_2)$ in the support of $h_1(\eta_1)h_2(\eta_2)$
one sees $|K_M^2|\lesssim I_1(J_{1,2}^\tau) $ and then proceeds as in the bound of $I_1(J_{1,2}^\tau)$.

Now we prove the case $(l,k)=(1,3)$, that is, we assume that the supports of $h_1$ and $h_3$ are
separated by $\tau$. In this case we have $|K^1_M|\le I(J_{1,3}^\tau)$ and $|K^2_M|\le I^1(J_{1,3}^\tau)$.
The triangle inequality yields $J_{1,3}^\tau \subset J_{1,2}^{\tau/2}\cup J_{2,3}^{\tau/2}$,
as subsets of $\R^4$, hence
 \bdm
    I(J_{1,3}^\tau) \le I(J_{1,2}^{\tau/2}) +  I(J_{2,3}^{\tau/2})
    \lesssim
    \frac{1}{\sqrt{\tau}} \prod_{j=1}^4 \|h_j\|_2,
 \edm
and similarly for $I^1(J_{1,3}^\tau)$.
This finishes the proof of the proposition.
\end{proof}

\section{Proof of exponential decay.} \label{sec:decay}

Let $f$ be a weak solution of the dispersion management equation.
Let
\beq
\|f\|_{\mu,\veps}:=\|e^{F_{\mu,\veps}(X)} f\|_2,
\eeq
with $F_{\mu,\veps}$ defined in \eqref{eq:Fmueps}.
The main step in our argument is to show that for some positive $\mu$, $\|f\|_{\mu,\veps}$ is
bounded in $\veps > 0$.

Fix $\tau>1$ and define, for an arbitrary function $f$,
$$f_\ll:=f\chi_{[-\tau/3,\tau/3]}, \quad  f_<:= f\chi_{[-\tau,\tau]},
\quad f_{>}:=f\chi_{[-\tau,\tau]^c}, \quad  f_\sim:=f_<-f_\ll.
$$

\begin{lem}\label{lem:flarge} Let $f$ be a weak solution of the dispersion management
equation for some $\omega>0$ with $\|f\|=1$.  Then
\begin{align*}
    \omega \|f_{>}\|_{\mu,\veps}
    &\lesssim
    \|f_{>}\|^3_{\mu,\veps} +e^{\mu\tau}\|f_{>}\|_{\mu,\veps}^2
    +\|f_>\|_{\mu,\veps}e^{2\mu\tau}\Big(\frac{1}{\sqrt{\tau}}
    +  \|f_\sim\|  \Big) +e^{3\mu\tau}\Big(\frac{1}{\sqrt{\tau}}  +\|f_\sim\| \Big).
\end{align*}
where the implicit constant does not depend on $\mu$,  $\veps$, and $\tau$.
\end{lem}
\begin{proof} Since $f$ is a weak solution of the dispersion management equation for some $\omega>0$, we have
$$\omega\la\varphi,f\ra=\calQ(\varphi,f,f,f), \text{ for any }\varphi\in L^2.
$$
Using this with $\varphi=e^{2F_{\mu,\veps}}f_>$, we obtain
\begin{align*}
\omega \|f_{>}\|_{\mu,\veps}^2 &= \calQ(e^{2F_{\mu,\veps}}f_{>},f,f,f)\\
&=\calQ_{\mu,\veps}(e^{F_{\mu,\veps}}f_{>},e^{F_{\mu,\veps}}f,e^{F_{\mu,\veps}}f,e^{F_{\mu,\veps}}f).
\end{align*}
Let $h:=e^{F_{\mu,\veps}}f$. Then
$$
\omega \|h_{>}\|^2=\calQ_{\mu,\veps}(h_{>},h,h,h).
$$
Writing $h=h_>+h_<$, and using the multilinearity of $\calQ_{\mu,\veps}$, we obtain
\begin{align} \label{eq:hbound}
\omega \|h_{>}\|^2&=\calQ_{\mu,\veps}(h_{>},h_>,h_>,h_>)+
\calQ_{\mu,\veps}(h_{>},h_<,h_<,h_<)\\
&+\calQ_{\mu,\veps}(h_{>},h_>,h_>,h_<)\nonumber
+\calQ_{\mu,\veps}(h_{>},h_>,h_<,h_<)\\
&+(\text{similar terms with permutations of the last three entrees})\nonumber
\end{align}
Note that by Theorem~\ref{thm:twistedQ}, we have
\begin{align*}
|\calQ_{\mu,\veps}(h_{>},h_>,h_>,h_>)|&\lesssim \|h_>\|^4,\\
|\calQ_{\mu,\veps}(h_{>},h_>,h_>,h_<)|&\lesssim \|h_>\|^3 \|h_<\|.
\end{align*}
To estimate the remaining terms, we will further split one of the $h_<$ they contain into $h_\ll+h_\sim$:
\begin{align*}
    |\calQ_{\mu,\veps}(h_{>},h_<,h_<,h_<)|
    &\leq
    |\calQ_{\mu,\veps}(h_{>},h_<,h_<,h_\ll)| +|\calQ_{\mu,\veps}(h_{>},h_<,h_<,h_\sim)|\\
    &\lesssim \frac{1}{\sqrt{\tau}}\|h_>\|\|h_<\|^2 \|h_\ll\| + \|h_>\|\|h_<\|^2 \|h_\sim\|,
\end{align*}
using Theorem~\ref{thm:twistedQ}, Theorem~\ref{thm:twistedQ-loc}, and the fact that the supports of $h_>$ and $h_\ll$ are separated by $2\tau/3$.
Similarly,
\begin{align*}
    |\calQ_{\mu,\veps}(h_{>},h_>,h_<,h_<)|
    &\leq
    |\calQ_{\mu,\veps}(h_{>},h_>,h_<,h_\ll)| +|\calQ_{\mu,\veps}(h_{>},h_>,h_<,h_\sim)|\\
    &\lesssim \frac{1}{\sqrt{\tau}}\|h_>\|^2\|h_<\| \|h_\ll\| + \|h_>\|^2\|h_<\| \|h_\sim\|
\end{align*}
Similar estimates hold for the  permurtations.
Using these estimates in \eqref{eq:hbound}, we obtain
\begin{align} \label{eq:hbound1}
    \omega \|h_{>}\|^2
    &\lesssim
    \|h_>\|^4 +\|h_>\|^3 \|h_<\|+\frac{1}{\sqrt{\tau}}\|h_>\|^2\|h_<\| \|h_\ll\|
    + \|h_>\|^2\|h_<\| \|h_\sim\| \\
    &+\frac{1}{\sqrt{\tau}}\|h_>\|\|h_<\|^2 \|h_\ll\| + \|h_>\|\|h_<\|^2 \|h_\sim\| \nonumber\\
\end{align}
Dividing both sides by $ \|h_>\|$ and using $h_<, h_\ll \leq e^{\mu\tau} f$,
$h_\sim\leq e^{\mu\tau}f_\sim$, and $\|f\|=1$, we obtain
\begin{align}
    \omega \|h_{>}\|
    &\lesssim
    \|h_>\|^3 +\|h_>\|^2 e^{\mu\tau} +\|h_>\|e^{2\mu\tau}\Big(\frac{1}{\sqrt{\tau}}
    +    \|f_\sim\|\Big)  + e^{3\mu\tau}\Big(\frac{1}{\sqrt{\tau}}  +   \|f_\sim\|\Big), \nonumber
\end{align}
which finishes the proof.
\end{proof}

\begin{proof}[Proof of Theorem~\ref{thm:main}]

Step 1. We will first determine $\tau > 1$ and we pick $\mu$ so that $e^{\mu\tau}=2$. We can rewrite the
bound from Lemma~\ref{lem:flarge} as
 (with the notation $\nu=\|h_>\|$)
\beq \label{eq:Gdef}
\big(\omega -\frac{C}{\sqrt{\tau}}-C\|f_\sim\|\big)\nu-C\nu^2-C\nu^3\leq C\Big(\frac{1}{\sqrt{\tau}}+\|f_\sim\|\Big).
\eeq

\noindent
Step 2. Let $G(\nu)=\frac{w}{2}\nu-C\nu^2-C\nu^3$. Let $\nu_\text{max}$ be the maxima of $G$ on $\R^+$.

\begin{figure}[htb]
\psfrag{a1}{$G_\text{max}$}
\psfrag{a0}{$G(\nu_0)$}
\psfrag{bm}{$\nu_\text{max}$}
\psfrag{b0}{$\nu_0$}
\psfrag{b1}{$\nu_1$}
\psfrag{0}{$0$}
\psfrag{text}{$[0,\nu_0]\cup[\nu_1,\infty)= G^{-1}\big([0,G(\nu_0)]\big)$}
\includegraphics[height=6cm,width=12cm]{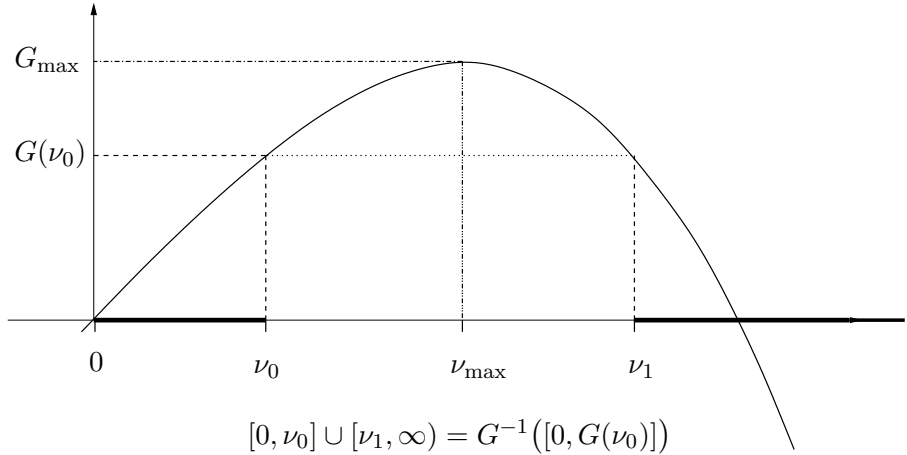}
\caption{Graph of $G(\nu)$ and the trapping region $G^{-1}\big([0,G(\nu_0)]\big)$.}
  \label{fig:G}
\end{figure}

\noindent
Step 3. Let  $\nu_0=\nu_\text{max}/2$, and pick $\tau>1$ so that
\begin{align*}
\text{i) } &  C\Big(\frac{1}{\sqrt{\tau}}+\|f_\sim\|\Big) \leq \min(\omega/2, G(\nu_0)),   \\
\text{ii) } & \|f_>\|\leq \nu_0/2.
\end{align*}
With this choice, we rewrite \eqref{eq:Gdef} as
\beq\label{eq:Gbound}
G(\|f_>\|_{\mu,\veps})\leq G(\nu_0),
\eeq
which is valid for any $\veps >0$. This is depicted in figure \ref{fig:G}.

\noindent
Step 4. Note that by ii) above and our choice of $\mu$ in step 1, we have
\beq\label{eq:start}
\|f_>\|_{\mu,1}\leq \|e^{\mu\frac{|x|}{1+|x|}}\|_\infty \|f_>\| \leq e^\mu\nu_0/2 < \nu_0.
\eeq
Finally since $\|f_>\|_{\mu,\veps}$ depends continuously on $\veps$ for $\veps>0$,
and \eqref{eq:start}, the inequality \eqref{eq:Gbound} shows that $\|f_>\|_{\mu,\veps}$ is
in the same connected component of $G^{-1}([0,G(\nu_0)])$, that is   $\|f_>\|_{\mu,\veps}\in [0, \nu_0]$ for all
$\veps>0$. This implies by monotone convergence theorem that
$$\|f_>\|_{\mu,0}=\sup_{\veps>0}\|f_>\|_{\mu,\veps}\leq \nu_0.$$
This shows that $e^{\mu |\cdot|} f \in L^2$.
With the obvious change of notation, a similar argument using Theorems~\ref{thm:twistedQ},
\ref{thm:twistedQ-loc} for $\tilde{\calQ}_{\mu,\veps}$ shows that
$e^{\tilde{\mu}|\cdot|}\widehat{f}\in L^2$,  for some $\tilde{\mu}>0$. Finally, the pointwise exponential bounds follows from the
one-dimensional Sobolev embedding theorem, or simply by the following
\begin{align*}
e^{\mu|x|}|f(x)|^2&=e^{\mu|x|} \Big|\int_x^\infty \frac{d}{ds}|f(s)|^2 ds\Big|\\
&\leq 2 \int_x^\infty e^{\mu|s|}|f(s)|   |f^\prime(s)| ds \leq 2 \|e^{\mu|\cdot|}f\|\,
\|f^\prime\|<\infty.
\end{align*}
Similarly one gets pointwise exponential decay of $\widehat{f}$.
\end{proof}



\textbf{Acknowledgements: }
It is  a pleasure to thank Vadim Zharnitsky
for instructive discussions on the dispersion management technique.

B. Erdo\smash{\u{g}}an and D. Hundertmark are partially supported by NSF grants
DMS-0600101 and DMS-0803120, respectively and Y.-R. Lee by the National Research Foundation of Korea (NRF)-grant 2009-0064945.
D. Hundertmark thanks  Max-Planck Institute for Physics of Complex Systems in Dresden
and the Max-Planck Institute for Mathematics in the Sciences in Leipzig for their warm
hospitality while part of this work was done.


\def\cprime{$'$}

\end{document}